\documentclass[conference,a4paper]{IEEEtran}
\usepackage{enumerate}
\usepackage[top=0.8in, bottom=1.5in, left=0.6in, right=0.6in]{geometry}
\usepackage{booktabs}
\usepackage{epstopdf}
\usepackage{cite}
\usepackage{amscd}
\usepackage{dsfont}
\usepackage{latexsym}
\usepackage{array}

\usepackage{tikz}\usetikzlibrary{calc}
\usepackage{tabularx}
\usepackage{epsfig}
\usepackage{graphicx,subfigure}
\usepackage{amsmath,mathrsfs}
\usepackage{amsthm}
\usepackage{amssymb}  
\usepackage{amsbsy}
\usepackage{color}
\usepackage{stfloats}
\usepackage{algorithm}
\usepackage{algpseudocode}
\usepackage{bm}
\usepackage{pifont}
\usepackage{accents}
\usepackage{multirow}

\theoremstyle{definition}
\newtheorem{rem}{Remark}
\theoremstyle{plain}
\newtheorem{proposition}{Proposition}

\ifCLASSINFOpdf
\else
\fi
\begin{document}
%
\title{Decoding Golay Codes and their Related Lattices: A PAC Code Perspective}

\author{
\IEEEauthorblockN{Yujun Ji$^1$, Ling Liu$^2$, Shanxiang Lyu$^3$, Chao Chen$^2$, Tao Dai$^1$, Baoming Bai$^2$}
\IEEEauthorblockA{$^1$College of Computer Science and Software Engineering, Shenzhen Univeristy, Shenzhen, China}
\IEEEauthorblockA{$^2$Guangzhou Institute of Technology, Xidian University, Guangzhou, China}
\IEEEauthorblockA{$^3$College of Cyber Security, Jinan University, Guangzhou, China}
\IEEEauthorblockA{liuling@xidian.edu.cn}
}


%


\maketitle

\begin{abstract}
In this work, we propose a decoding method of Golay codes from the perspective of Polarization Adjusted Convolutional (PAC) codes. By invoking Forney's cubing construction of Golay codes and their generators $G^*(8,7)/(8,4)$, we found different construction methods of Golay codes from PAC codes, which result in an efficient parallel list decoding algorithm with near-maximum likelihood performance. Compared with existing methods, our method can get rid of index permutation and codeword puncturing. Using the new decoding method, some related lattices, such as Leech lattice $\Lambda_{24}$ and its principal sublattice $H_{24}$, can be also decoded efficiently.
\end{abstract}


%
\IEEEpeerreviewmaketitle

\section{Introduction}

The Golay code is among the pioneering codes of classic error correction theory, tracing its origins to the late 1940s \cite{golay1949notes}. As an example of a perfect binary error-correcting code, the (23, 12, 7) Golay code features the notable capability of correcting up to three errors. The extended (24, 12, 8) Golay code can be constructed by appending a single-parity-check bit to the length-23 Golay code. This extended code is the unique (24, 12, 8) binary code; it is also an extended perfect doubly-even self-dual code with a code rate of 1/2. Due to its elegant structure and practical properties, the Golay code has garnered significant interest since its inception, leading to the proposal of various construction approaches over the decades \cite{peng2006construction}. A range of decoders have been developed for Golay codes, including algebraic decoders \cite{macwilliams1977theory} and soft-decision decoders \cite{sarangi2014efficient,adde2011near,sorger2000star}. 

The Leech lattice $\Lambda_{24}$ is a 24-dimensional lattice with several equivalent definitions \cite{conway1999sphere,lepowsky1982e8,niemeier1973definite,sloane1980note,sloane2003tables}. 
For example, it can be constructed by applying Construction B \cite{leech1971sphere,sloane1977binary} to the extended Golay code, followed by doubling the number of points. In \cite{forney1988coset}, Forney gave the cubing construction of the Golay code and the Leech lattice, using the partition chains (8, 7, 2)/(8, 4, 4)/(8, 1, 8) and $E_8/RE_8/2E_8$ respectively, where the $E_8$ is a Barnes-Wall lattice with dimension 8. 
The Leech lattice $\Lambda_{24}$ has a more remarkable status among lattices than the Golay code in binary codes. 
Due to their exceptional geometric properties, Leech lattices are generally used for modulation shaping and vector quantization. Therefore, the development of fast decoding algorithms for them is critical. 


Polar codes, introduced by Ar{\i}kan \cite{arikan2009channel}, are the first provably capacity-achieving codes with explicit construction and low encoding and decoding complexity.
Their capacity-achieving principle is channel polarization: this process converts a set of identical channels into synthesized channels that asymptotically tend to be either perfectly reliable or completely unreliable as code length approaches infinity.
In their original design, polar codes have lengths restricted to powers of two; other lengths can be achieved via puncturing or shortening \cite{bioglio2017low}. 
Polar codes with the original 2-by-2 kernel cannot achieve the optimal scaling law for finite block lengths due to the suboptimal polarization rate (i.e., $\beta=1/2$), which can be improved by adopting more complex kernels. 
Recently, Ar{\i}kan proposed the polarization-adjusted convolutional (PAC) codes \cite{arikan2019sequential} to improve the performance of polar codes in short-to-medium block lengths. 
PAC codes embed an outer convolutional coding block before the polar transform, a design intended to avoid capacity loss in synthesized channels. 
By employing sequential decoding \cite{arikan2019sequential} or list decoding \cite{yao2021list,rowshan2021polarization}, PAC codes demonstrate superior performance over polar codes for intermediate block lengths. 


In this paper, we propose an efficient parallel decoding method of Golay codes from the perspective of PAC codes. The Golay codes are constructed by Forney's cubing construction and generators $G^*(8,7)/(8,4)$ \cite{forney1988coset}. We find that the generator matrix of Golay codes can be constructed from several standard polar generator matrices and carefully designed convolutional matrices $T$, thereby enabling an efficient parallel decoding framework in Sect. \ref{sec: decodiong framework}. Compared to \cite{bioglio2018polar}, our method does not need index permutation and puncturing, and it is compatible with standard successive cancellation list (SCL) decoding of PAC codes. 
We also notice that in \cite{lin2020transformation}, Golay codes were converted to a pre-transformed polar codes using a permutation matrix, which was usually derived from heuristic searches. Our method is deterministic and it can achieve performance limit with smaller list, as shown in Sect. \ref{sec: simulations}. 
Furthermore, based on Forney's code formulas, our proposed decoding method for Golay codes can be extended to some classical lattices like $\Lambda_{24}$ and its principal sublattice $H_{24}$ in Sect. \ref{sec: two-lattices}. 
The decoding method can find its application to post quantum cryptography, as recent advances \cite{lyu2024lattice,liu2025lattice,liu2025compact} have employed small dimensional lattice codes for error correction.


The rest of the paper is organized as follows: Sect. II gives preliminaries of polar codes and PAC codes. The constructions of Golay codes from PAC codes are described in Sect. III. A novel decoding algorithm is also provided. In Sect. IV, we use the new decoding algortihm to decode some related lattices. The paper is concluded in Sect. V.

All random variables (RVs) are denoted by capital letters, and their realizations are denoted by the corresponding lowercase letters. $\left [ N \right ] $ denotes the set $\left \{ 1,2,\dots,N  \right \} $. We use the notation $x_{i}^j(i\le j)$ as a shorthand for the vector $(x_i,..., x_j)$, which is a realization of the RVs $X_{i}^j=(X_i,..., X_j)$. 
The vector ($x_i:i\in\mathcal{A}$) is denoted by $x_{\mathcal{A}}$. For a set $\mathcal{A}$, $\mathcal{A}^c$ denotes its complement and $|\mathcal{A}|$ its cardinality. The capacity of channel $W$ is denoted by $C(W)$. 

\section{Preliminaries of Polar Codes and PAC Codes}\label{sec:background}

\subsection{Polar Codes}

The generator matrix of polar codes is defined as 
\begin{align}
    G_N^{p}=F_2^{\otimes n}\times B_N,
\end{align}
where $N=2^n$, $F_2=\left[\begin{smallmatrix}1&0\\1&1\end{smallmatrix}\right]$, $\otimes$ denotes the Kronecker product, and $B_N$ is the bit-reverse permutation matrix \cite{arikan2009channel}. 
For polar encoding, $u_{\mathcal{A}^c}$ and $u_\mathcal{A}$ are combined into the vector $u_1^N = [u_{\mathcal{A}^c}, u_\mathcal{A}]$, where $u_\mathcal{A}$ represents the information bits and $u_{\mathcal{A}^c}$ represents the frozen bits. Then, the codeword is given by $x_1^N = u_1^N G_N^{p}$.

Let $W$ be a binary input memoryless symmetric channel (BMSC). 
For channel polarization, $N$ independent copies of a given $W$ are combined into a vector channel $W_N$, which is split into $N$ subchannels  $W_N^{(i)}$, $1 \le i \le N$. 
The codeword $x_1^N$ is transmitted over a vector channel $W_N$ with outputs $y_{1}^N$. 
As $N$ increases to infinity, almost all subchannels have capacity close to 0 or 1, and the proportion of the subchannels with capacity close to 1 approaches $C(W)$. For polar decoding, one can use a successive cancellation (SC) decoder to estimate  $\hat{u}_\mathcal{A}$, where
\begin{flalign}
    \hat{u}_i = 
    \begin{cases}
        0, & \text{if } i \in \mathcal{A}^c \\
        \arg\max_{u\in\{0,1\}} P(u\mid y_1^N,\hat{u}_1^{i-1}), & \text{if } i\in\mathcal{A}
    \end{cases} .
\end{flalign}
The SCL decoder is an enhancement of the SC decoder, which expands its single decoding path to $L$ paths. 
As $L$ increases, the decoding success probability improves. 

To construct polar codes with arbitrary blocklength $N$ (not restricted to powers of two), a multi-kernel polarization technique was proposed in \cite{gabry2017multi}. This technique employs multiple kernels of varying sizes in combination to perform polarization. 
Variations in kernel selection and combination methods can result in different generator matrices.

\subsection{PAC Codes}

The generator matrix of PAC codes 
\begin{align}
    G_N^{P} = TG_N^{p}
\end{align}
is obtained by combining the polar generator matrix $G_N^{p}$ with an upper triangular convolutional matrix 
\begin{align*}
    T=\begin{bmatrix}
 c_1 & \cdots  & c_t & 0 & \cdots &0 \\
 0 & c_1 & \cdots & c_t & \ddots  & \vdots \\
 \vdots & \ddots & \ddots & \cdots & \ddots & 0\\
  \vdots&  & \ddots & \ddots & \cdots & c_t\\
 \vdots &  &  & \ddots & \ddots & \vdots\\
 0& \cdots &\cdots  & \cdots& 0 &c_1
\end{bmatrix},
\end{align*}
constructed from a convolutional vector 
$c_1^t$ where $c_1=1$. 
In fact, a more general form of $T$ can be used, where $c_1^t$ can be different for different rows. 
For PAC encoding, now the index partition is performed on $V^{[N]}$ instead of $U^{[N]}$, i.e., $v_1^N = [v_{\mathcal{A}^c}, v_\mathcal{A}]$. The codeword is given by $x_1^N = v_1^N G_N^{P}$. 


Initially, Ar{\i}kan adopted sequential decoding for PAC codes, which achieved excellent performance \cite{arikan2019sequential}. Subsequent studies have confirmed that SCL decoders can produce comparable performance when the list size is set to a sufficiently high value \cite{rowshan2021polarization}.



\section{Constructing Golay Codes from PAC Codes}

The construction of Golay codes follows \cite{peng2006construction}, which is an extension of the classical approach proposed in \cite{macwilliams1977theory}. 
The code is formed from a (3, 2, 2) single-parity-check code and a (3, 1, 3) repetition code (generator matrices $S$, $R$) as well as two (8, 4, 4) linear block codes with generator matrices $G_8$ and $G_8^{'}$. In fact, if all weight-4 codewords of $G_8$ and $G_8^{'}$ are distinct, the code generated by the matrix
\begin{align}
    \hat{G}=\begin{bmatrix}
S \otimes G_8 \\
R\otimes G_8^{'}
\end{bmatrix}
=\begin{bmatrix}
 G_8 & 0 & G_8\\
0  &G_8  & G_8\\
G_8^{'}  &G_8^{'}  &G_8^{'}
\end{bmatrix}\label{golay-construction}
\end{align}
is a (24, 12, 8) extended Golay code. 

\subsection{Motivation} \label{subsec: motivation}
The matrix $G_8$ is adopted as the generator matrix of the (8, 4, 4) Reed-Muller (RM) code, derived by freezing the input bits corresponding to rows of $F_8=F_2^{\otimes3}$ with Hamming weight smaller than 4 \cite{hadi2016enhancing}. 
We define $\pi $ as a permutation of the index set $[8]=\{1,2,\cdots,8\}$. The permutation can describe the relationship between the columns of the matrices $G_8$ and $G_8^{'}$, where their connection is given by:
\begin{align*}
    G_8^{'}(:, \pi) = G_8 \ \ \text{and} \ \  G_8(:, \pi^{-1}) = G_8^{'}.
\end{align*}
In \cite{bioglio2018polar}, the adopted permutation is $\pi_1 = [2\ 5\ 4\ 3\ 6\ 7\ 1\ 8]$, giving 
\begin{align}
    G_8^{'(1)}=\begin{bmatrix}
\mathbf{g}_1^{(1)} \\
\mathbf{g}_2^{(1)} \\
\mathbf{g}_3^{(1)} \\
\mathbf{g}_4^{(1)}
\end{bmatrix}
=\begin{bmatrix}
 0 & 1 &1  &1  & 1 & 0 & 0 &0 \\
 0 & 1 & 0 & 0 & 1 & 1 & 1 & 0\\
 1 & 1 & 0 & 1 & 0 & 1 & 0 & 0\\
 1 & 1 & 1 & 1 & 1 & 1 & 1 & 1
\end{bmatrix}.
\end{align}

The authors have not yet found a corresponding matrix $T$ to obtain $G_8^{'(1)}$ from $F_8$ with relatively short span, thus necessitating the introduction of index permutation and codeword puncturing. 
Since $\mathbf{g}_4^{(1)}$ is unaffected by the permutation, it will not be further discussed hereafter. 
In addition, $\pi_2 = [4\ 5\ 2\ 3\ 6\ 1\ 7\ 8]$ was employed in an earlier work by Forney \cite{forney1988coset}, with the generators adopted given by 
\begin{align*}
    \mathbf{g}_1^{(2)}&=[0\ 1 \ 1 \ 1 \ 1 \ 0 \ 0\ 0] \\
    \mathbf{g}_2^{(2)}&=[1\ 0 \ 0 \ 1 \ 1 \ 1 \ 0\ 0] \\
    \mathbf{g}_3^{(2)}&=[0\ 1 \ 0 \ 1 \ 0 \ 1 \ 1\ 0].
\end{align*}
We redefine a set of generators and illustrate how to derive it from the rows of $F_8$: 
\begin{align*}
    \mathbf{g}_1^{(3)}&=\mathbf{g}_1^{(2)} =\mathbf{f}_4 \oplus \mathbf{f}_5\\
    \mathbf{g}_2^{(3)}&=\mathbf{g}_2^{(2)}=\mathbf{f}_3 \oplus \mathbf{f}_4 \oplus \mathbf{f}_6  \\
    \mathbf{g}_3^{(3)}&=\mathbf{g}_2^{(2)} \oplus \mathbf{g}_3^{(2)}=\mathbf{f}_2 \oplus \mathbf{f}_3 \oplus \mathbf{f}_7,
\end{align*}
where $\mathbf{f}_i \ (i=1,2,\dots,8)$ denotes the $i$-th row vector of $F_8$, and $\mathbf{g}_3^{(3)}$ is the XOR of $\mathbf{g}_2^{(2)}$ and $\mathbf{g}_3^{(2)}$. 
The corresponding permutation for this set of generators is $\pi_3 = [5\ 4\ 2\ 3\ 1\ 6\ 7\ 8]$.

\begin{proposition}
For the generator matrix $G_8^{'(3)}$ derived from the permutation $\pi_3 = [5\ 4\ 2\ 3\ 1\ 6\ 7\ 8]$, there exists a corresponding PAC code. Its relevant parameters are given by
\begin{align*}
    \mathcal{A} = \{2,3,4,8\}, \quad G_{8}^{p}=F_8,
\end{align*}
and
\begin{align}
    T=\begin{bmatrix}
    1 & 0 & 0 & 0 & 0 & 0 & 0 & 0\\
    0 & 1 & 1 & 0 & 0 & 0 & 1 & 0\\
    0 & 0 & 1 & 1 & 0 & 1 & 0 & 0\\
    0 & 0 & 0 & 1 & 1 & 0 & 0 & 0\\
    0 & 0 & 0 & 0 & 1 & 0 & 0 & 0\\
    0 & 0 & 0 & 0 & 0 & 1 & 0 & 0\\
    0 & 0 & 0 & 0 & 0 & 0 & 1 & 0\\
    0 & 0 & 0 & 0 & 0 & 0 & 0 & 1
    \end{bmatrix}.
    \label{eq:matrix_T}
\end{align}
\label{prop:1}
\end{proposition}




\begin{rem}
    In fact, for any feasible permutation $\pi$, if the corresponding PAC codes (i.e., the corresponding parameters $\mathcal{A}$ and $T$) can be obtained, then our framework is effective. 
    Furthermore, the parity-check constraint on $T$ proposed in Proposition \ref{prop:1} is relatively simple, which makes its performance analysis convenient in future work.
\end{rem}

Motivated by Proposition \ref{prop:1}, we found three different construction methods for Golay codes derived from PAC codes, without index permutation and puncturing.

\subsection{Three Methods based on Different Kernels}
\subsubsection{$3\times3$ Kernel 1} Consider Forney's cube kernel 
\begin{align}
    F_3^{(1)}=\begin{bmatrix}
  1& 1 & 0\\
 1 & 0 & 1\\
 1 &  1&1
\end{bmatrix},
\end{align}
and then 
\begin{align*}
    G_{24}^{p,(1)} = F_3^{(1)}\otimes F_8=\begin{bmatrix}
  F_8& F_8 & 0\\
 F_8 & 0 & F_8\\
 F_8 &  F_8&F_8
\end{bmatrix}.
\end{align*}
It can be directly verified that the minimum distance of $F_3^{(1)}\otimes G_8$ is only 4, which is not optimal. 
To get a (24, 12, 8) extended Golay code, we should contruct a matrix
\begin{align}
    \hat{G}_{24}=\begin{bmatrix}
 G_8 & 0 & G_8\\
0  &G_8  & G_8\\
G_8^{'}  &G_8^{'}  &G_8^{'}
\end{bmatrix}\label{eq:G24},
\end{align}
where $G_8$ is a (8, 4, 4) linear block code and $G_8^{'}$ is a permutated version of $G_8$. 
In subsequent discussions, $G_8$ denotes the matrix introduced in Section \ref{subsec: motivation}, and $G_8^{'}$ denotes $G_8^{'(3)}$. Both matrices are given as follows: 
\begin{align}
    G_8=\begin{bmatrix}
 1 &1  & 1 & 1 & 0 & 0 & 0 & 0\\
 1 & 1 & 0 & 0 & 1 & 1 & 0 &0 \\
 1 & 0 & 1 & 0 &1  & 0 & 1 &0 \\
 1 & 1 &1  & 1 & 1 & 1 & 1 &1
\end{bmatrix},
\end{align}
\begin{align}
    G_8^{'}=\begin{bmatrix}
 0 &1  & 1 & 1 & 1 & 0 & 0 & 0\\
 1 & 0 & 0 & 1 & 1 & 1 & 0 &0 \\
 1 & 1 & 0 & 0 &1  & 0 & 1 &0 \\
 1 & 1 &1  & 1 & 1 & 1 & 1 &1
\end{bmatrix}.
\end{align}

\begin{proposition}
Eight-bit block-wise column permutation of $\hat{G}_{24}$ does not change its codebook.
\label{prop:2}
\end{proposition}

\begin{proof}
    Without loss of generality, assume that the matrix $\hat{G}_{24}^{'}$ is obtained by performing a block-wise column permutation on $\hat{G}_{24}$  in 8-bit column blocks, which is given by 
    \begin{align*}
        \hat{G}_{24}^{'}=\begin{bmatrix}
    0  & G_8 & G_8\\
    G_8  &G_8  & 0\\
    G_8^{'}  &G_8^{'}  &G_8^{'}
    \end{bmatrix}.
    \end{align*}
    For the coding scheme with $\hat{G}_{24}$ as the generator matrix, the codeword is calculated by
    \begin{align*}
        [\mathbf{x}_1 , \mathbf{x}_2, \mathbf{x}_3]=[\mathbf{a}_1 , \mathbf{a}_2, \mathbf{a}_3] \times \hat{G}_{24},
    \end{align*}
    where $\mathbf{x}_i$ and $\mathbf{a}_i \ (i=1,2,3)$ are 8-dimensional and 4-dimensional row vectors, respectively. Furthermore, 
    \begin{align*}
        \left\{
        \begin{array}{l}  
         \mathbf{x}_1=\mathbf{a}_1G_8+\mathbf{a}_3G_8^{'} \\
         \mathbf{x}_2=\mathbf{a}_2G_8+\mathbf{a}_3G_8^{'} \\
         \mathbf{x}_3=\mathbf{a}_1G_8+\mathbf{a}_2G_8+\mathbf{a}_3G_8^{'}
        \end{array}.
        \right.
    \end{align*}
    The codeword $[\mathbf{x}_1 , \mathbf{x}_2, \mathbf{x}_3]$ corresponds to the coding scheme that takes $\hat{G}_{24}^{'}$ as the generator matrix, satisfying
    \begin{align*}
        [\mathbf{x}_1 , \mathbf{x}_2, \mathbf{x}_3]=[\mathbf{a}_1+\mathbf{a}_2 , \mathbf{a}_1, \mathbf{a}_3] \times \hat{G}_{24}^{'}.
    \end{align*}
    Since the mapping between the uncoded sequences is invertible, it follows that the codebooks generated by $\hat{G}_{24}$ and $\hat{G}_{24}^{'}$ are identical.
\end{proof}

To construct Golay codes from PAC codes, which amounts to finding an information set $\mathcal{A}_1$ and an upper triangular matrix $T_1=(T_{i,j}^{(1)})_{1\le i,j \le 24}$ that satisfy the following equation:
\begin{align}
v_{\mathcal{A}_1}\times \small\begin{bmatrix}
  G_8& G_8 & 0\\
 G_8 & 0 & G_8\\
 G_8^{'} &  G_8^{'}&G_8^{'}
\end{bmatrix}=[v_{\mathcal{A}_1}, \mathbf{0}]\times T_1G_{24}^{p,(1)}.
\label{eq:kernel-1-golay_pac}
\end{align}
In conjunction with Proposition \ref{prop:1}, we can derive the form of $\mathcal{A}_1$ and $T_1$ as follows: 
\begin{equation*}
    \mathcal{A}_1 = \{4,6,7,8,12,14,15,16,18,19,20,24\},
\end{equation*}
\begin{align*}
    T_1 = \begin{bmatrix}
    I_8  & 0 & 0\\
    0  & I_8  & 0\\
     0 & 0 &T
    \end{bmatrix},
\end{align*}
where $I_8$ is the 8-order identity matrix, and $T$ is the matrix given in \eqref{eq:matrix_T}. 

\subsubsection{$3\times3$ Kernel 2} Consider kernel 
\begin{align}
    F_3^{(2)}=\begin{bmatrix}
  1& 0 & 0\\
 1 & 1 & 0\\
 0 &  1&1
\end{bmatrix},
\end{align}
and then
\begin{align*}
    G_{24}^{p,(2)} = F_3^{(2)}\otimes F_8=\begin{bmatrix}
  F_8& 0 & 0\\
 F_8 & F_8 & 0\\
 0 &  F_8&F_8
\end{bmatrix}.
\end{align*}
For another row-permutated version of $\hat{G}_{24}$, we have
\begin{align}
v_{\mathcal{A}_2}\times \small\begin{bmatrix}
  G_8'& G_8' & G_8'\\
 G_8 & G_8 &0 \\
 0 &  G_8&G_8
\end{bmatrix}=[v_{\mathcal{A}_2}, \mathbf{0}]\times T_2G_{24}^{p,(2)},
\label{eq:kernel-2-golay_pac}
\end{align}
and the corresponding results are
\begin{equation*}
    \mathcal{A}_2 = \{2,3,4,8,12,14,15,16,20,22,23,24\}
\end{equation*}
and
\begin{align*}
    T_2 = \begin{bmatrix}
    T  & 0 & T\\
    0  & I_8  & 0\\
     0 & 0 &I_8
    \end{bmatrix}.
\end{align*}
The selection of $\mathcal{A}_2$ exactly corresponds to the RM profile under $T_2G_{24}^{p,(2)}$.


\subsubsection{$3\times3$ Kernel 3} Consider kernel 
\begin{align}
    F_3^{(3)}=\begin{bmatrix}
  1& 0 & 0\\
 1 & 1 & 0\\
 1 &  1&1
\end{bmatrix},
\end{align}
and then
\begin{align*}
    G_{24}^{p,(3)} = F_3^{(3)}\otimes F_8=\begin{bmatrix}
  F_8& 0 & 0\\
 F_8 & F_8 & 0\\
 F_8 &  F_8&F_8
\end{bmatrix}.
\end{align*}
For another block-structured version of $\hat{G}_{24}$, we have
\begin{align}
v_{\mathcal{A}_3}\times \small\begin{bmatrix}
0 &  G_8&G_8 \\
G_8 & G_8 &0  \\
  G_8'& G_8' & G_8'
\end{bmatrix}=[v_{\mathcal{A}_3}, \mathbf{0}]\times T_3G_{24}^{p,(3)},
\label{eq:kernel-3-golay_pac}
\end{align}
and the corresponding results are
\begin{equation*}
    \mathcal{A}_3 = \{4,6,7,8,12,14,15,16,18,19,20,24\}
\end{equation*}
and
\begin{align*}
    T_3 = \begin{bmatrix}
    I_8  & 0 & I_8\\
    0  & I_8  & 0\\
     0 & 0 &T
    \end{bmatrix}.
\end{align*}



\subsection{A Combined Parallel Decoding Method}
\label{sec: decodiong framework}
To elaborate on the design principle of the combined parallel decoding method, we compare Eqs. \eqref{eq:kernel-1-golay_pac}, \eqref{eq:kernel-2-golay_pac} and \eqref{eq:kernel-3-golay_pac}. 
For the left hand sides, according to Proposition \ref{prop:2}, the corresponding Golay codes are identical. For the right hand sides, these Golay codes correspond to three PAC codes with distinct structures. 
Each of these structures is decoded by a dedicated SCL decoder; the three decoders operate in parallel, and the optimal path from their outputs is selected as the final decoding result.

The parallel decoding framework proposed herein is illustrated in Fig. \ref{fig:new-decoder}. The input corresponds to the log-likelihood ratio (LLR) sequence of the received sequence $y_1^{24}$:
\begin{itemize}
    \item Level-1 is shared among the three SCL decoders.  
    \item At Level-2, the three SCL decoders execute parallel decoding operations, relying on their respective kernels. 
    \item Level-3 undertakes the task of selecting the optimal path from the $3L$ candidate paths generated by Level-2. 
\end{itemize}


\begin{figure}[ht]
    \centering
    \includegraphics[width=7cm]{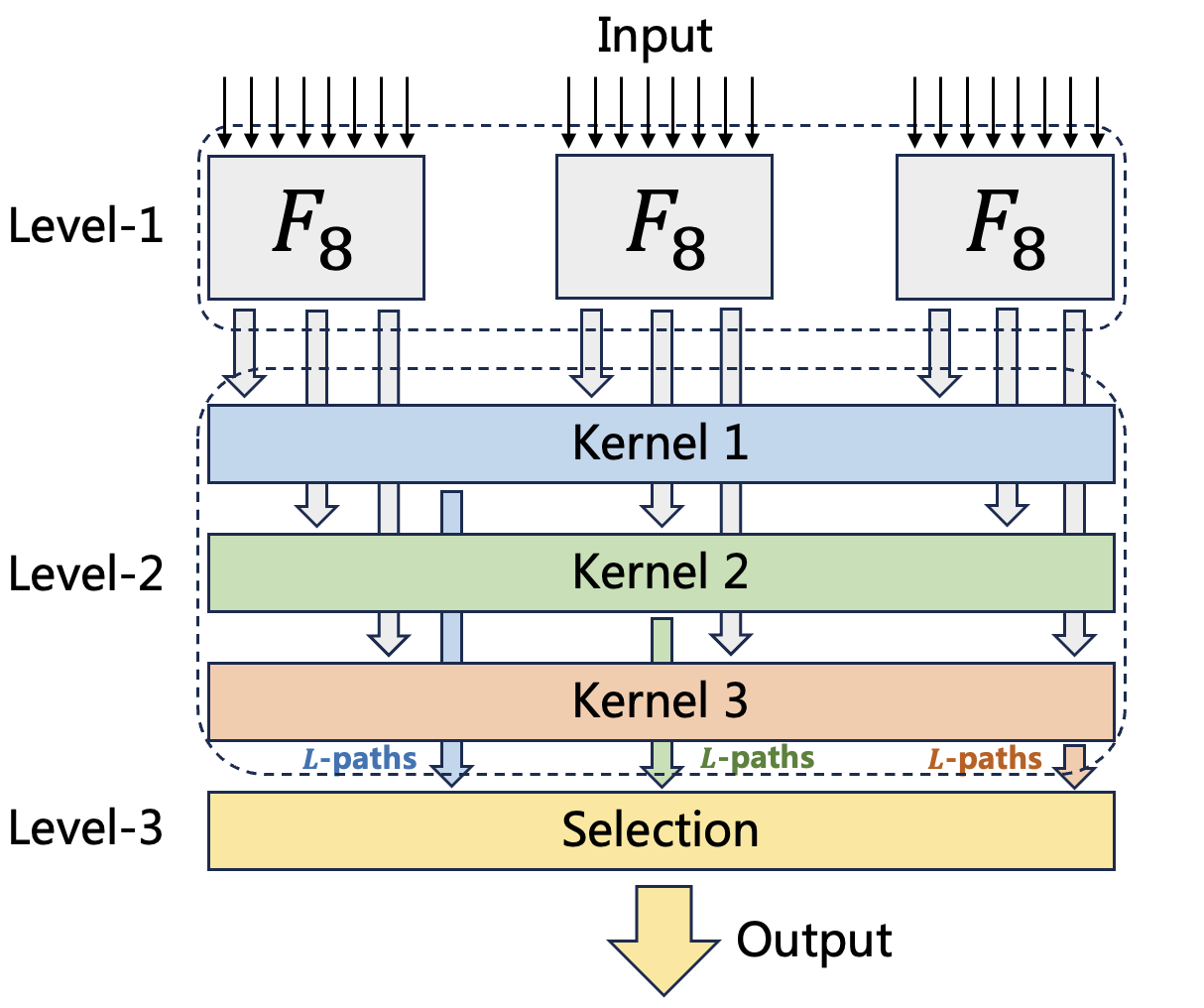}
    \caption{Parallel decoding framework of Golay codes.}
    \label{fig:new-decoder}
\end{figure}

\subsection{Simulation Results and Comparison}
\label{sec: simulations}

Fig. \ref{fig:bler_fig_K1_K2_K3} demonstrates the block error rate (BLER) performance of the three different PAC decoders based on the extended Golay code for transmission over an AWGN channel with BPSK modulation. 
In this figure, $K_i$ denotes the scheme employing the $i$-th 3×3 kernel. When the list size is small, the decoding error rate of $K_2$ is significantly higher than that of the other two kernels; 
this is because the convolutional span of $T_2$ under $K_2$ is the largest, thus imposing a higher demand for the list size. Moreover, there exist certain discrepancies in the performance of different $K_i$, which stems from the fact that distinct $T_i$ matrices exhibit varying degrees of coverage for different bit positions. For this reason, it is highly necessary to adopt a combined parallel decoding to cover all bit positions as comprehensively as possible. 
Once the list size exceeds 4, the three kernels achieve comparable performance.

Fig. \ref{fig:Bler_K1_Para_ML} demonstrates the BLER performance of the two distinct PAC code decoders and the maximum-likelihood (ML) decoder. When the list size is small, the proposed parallel decoding method achieves a noticeable performance gain; when the list size reaches 8, the proposed decoding algorithm attains near-maximum likelihood performance, which is much smaller than the list size of 64 required in \cite{lin2020transformation}; meanwhile, we no longer require the index permutation and puncturing operations introduced in \cite{bioglio2018polar}.

\begin{figure}[ht]
    \centering
    \includegraphics[width=7.5cm]{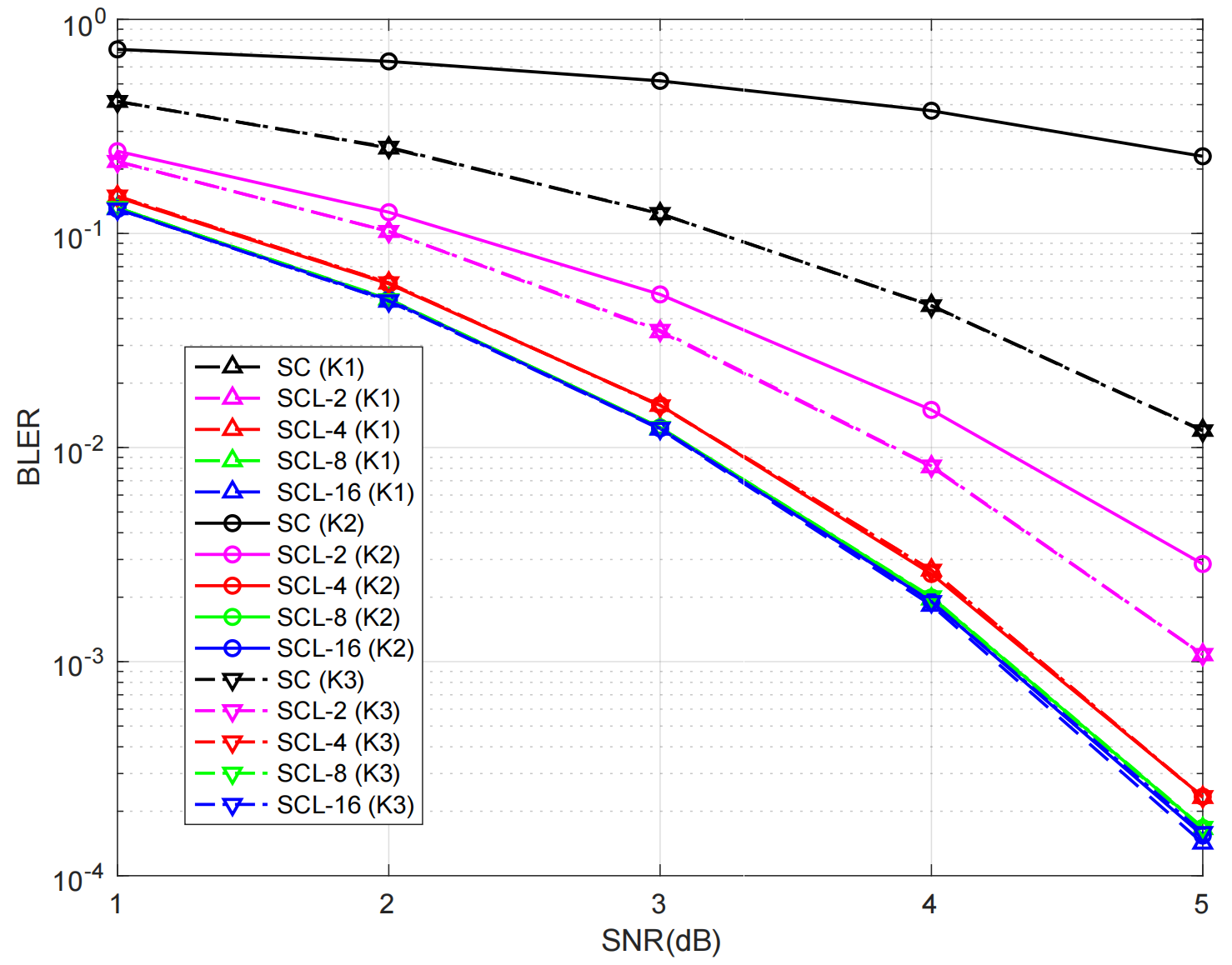}
    \caption{BLER performance of the (24, 12, 8) extended Golay code under SCL decoding with three kernels.}
    \label{fig:bler_fig_K1_K2_K3}
\end{figure}

\begin{figure}[ht]
    \centering
    \includegraphics[width=7.5cm]{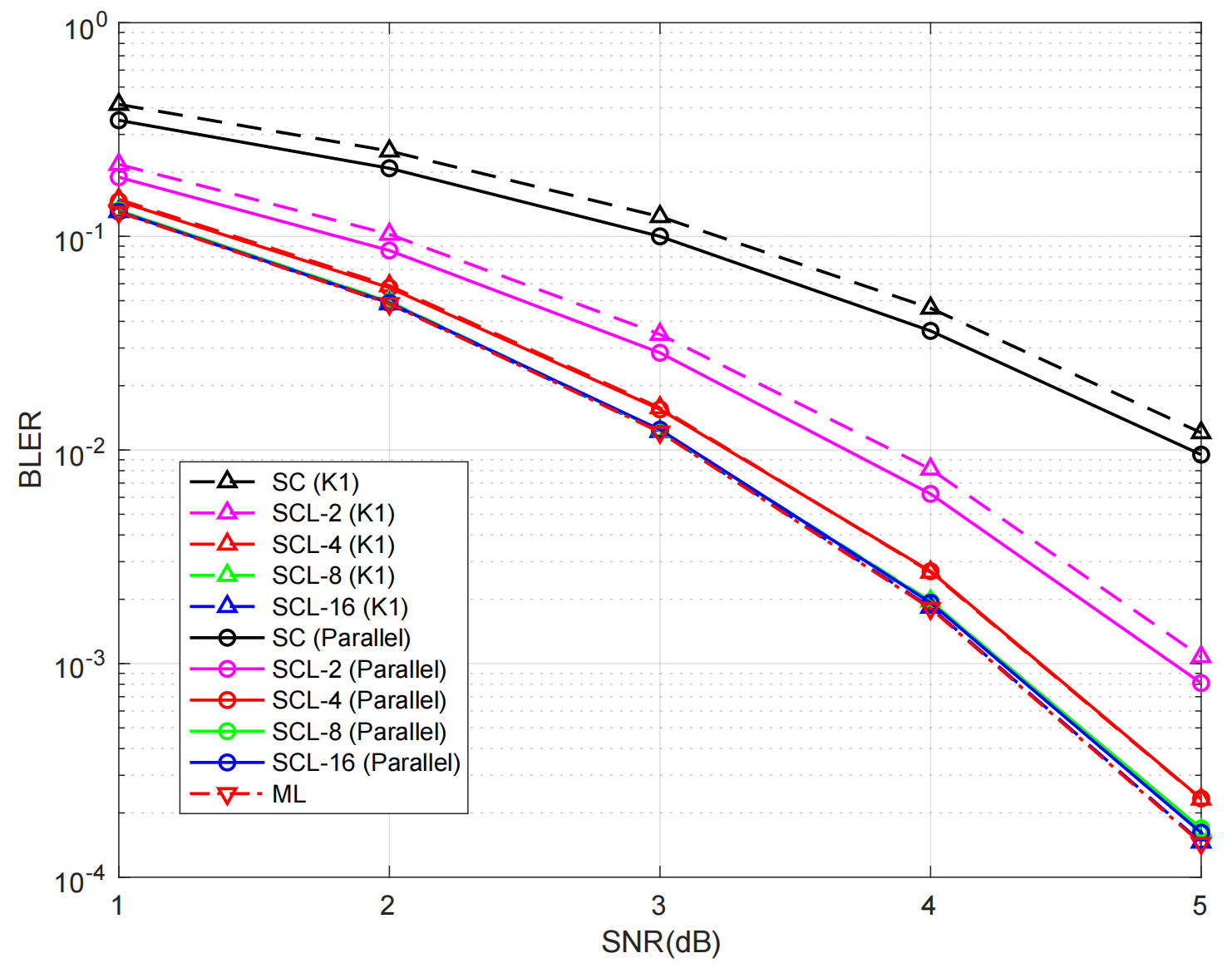}
    \caption{BLER performance of the (24, 12, 8) extended Golay code under SCL decoding with Kernel 1, parallel SCL decoding and ML decoding.}
    \label{fig:Bler_K1_Para_ML}
\end{figure}

\section{Multilevel Decoding of $\Lambda_{24}$ and $H_{24}$} \label{sec: two-lattices}

\subsection{Multilevel Decoding of Binary Lattices}

A lattice $\Lambda \subseteq \mathbb{R}^N$ is defined as a discrete additive subgroup of $\mathbb{R}^N$, which can be represented as
\begin{align}
    \Lambda = \{\boldsymbol{\lambda} = B\mathbf{x}=\sum_{i=1}^{N}x_i\mathbf{g}_i \mid\mathbf{x}\in\mathbb{Z}^N\},
\end{align}
where $B = [\mathbf{g}_1, \mathbf{g}_2,\dots,\mathbf{g}_N]$ denotes the generator matrix of $\Lambda$, $\mathbf{g}_i$ is an $N\times 1$ column vector, and these $N$ column vectors are linearly independent.


A lattice can be constructed from binary coset codes using Construction D. To do so, one needs a partition chain based on sublattices. The simplest is $\mathbb{Z}/2\mathbb{Z}/4\mathbb{Z}/\cdots$. More details for the Construction D lattices can be found in \cite{forney2000sphere, liu2018construction, liu2025revisit}.

We focus on the scenario where an $N$-dimensional Gaussian noise vector with zero mean and variance $\sigma^2$ per dimension falls outside the Voronoi region $\mathcal{V}(\Lambda)$ of a lattice $\Lambda$. Given the lattice $\Lambda$, we define a $\Lambda$-periodic function as
\begin{equation*}
    f_{\sigma,\Lambda}(\mathbf{x})=\sum_{\boldsymbol{\lambda} \in \Lambda}f_{\sigma,\boldsymbol{\lambda}}(\mathbf{x}) =\frac{1}{(\sqrt{2\pi}\sigma)^N} \sum_{\boldsymbol{\lambda} \in \Lambda} e^{-\frac{\left \| \mathbf{x}-\boldsymbol{\lambda} \right \|^2 }{2\sigma^2} }
\end{equation*}
for $\mathbf{x}\in\mathbb{R}^N$. This function represents the probability density function of Gaussian noise after a mod-$\mathcal{V}(\Lambda)$ operation.

It can be seen that Leech lattices and their principal sublattices are built by repetition codes, Golay codes and parity-check codes, the fast ML decoding algorithms of which are now available. It is thus promising to decode them separately using multilevel decoding while ensuring favorable performance. 
The framework for multilevel decoding is illustrated in Fig. \ref{fig:multilevel-decoding}. Specifically, the received signals are subjected to a modulo-2 operation at the front end of each decoder. At the output end, these signals are subtracted by the recovered codeword and then scaled down by a factor of two, before being forwarded to the subsequent decoding level. It is worth noting that the codeword corresponding to each level is transmitted over a mod-$\Lambda$ channel, which is essentially an additive white Gaussian noise (AWGN) channel integrated with a mod-$\mathcal{V}(\Lambda)$ operator at the receiver front end \cite{forney2000sphere}.

\begin{figure}[ht]
    \centering
    \includegraphics[width=8cm]{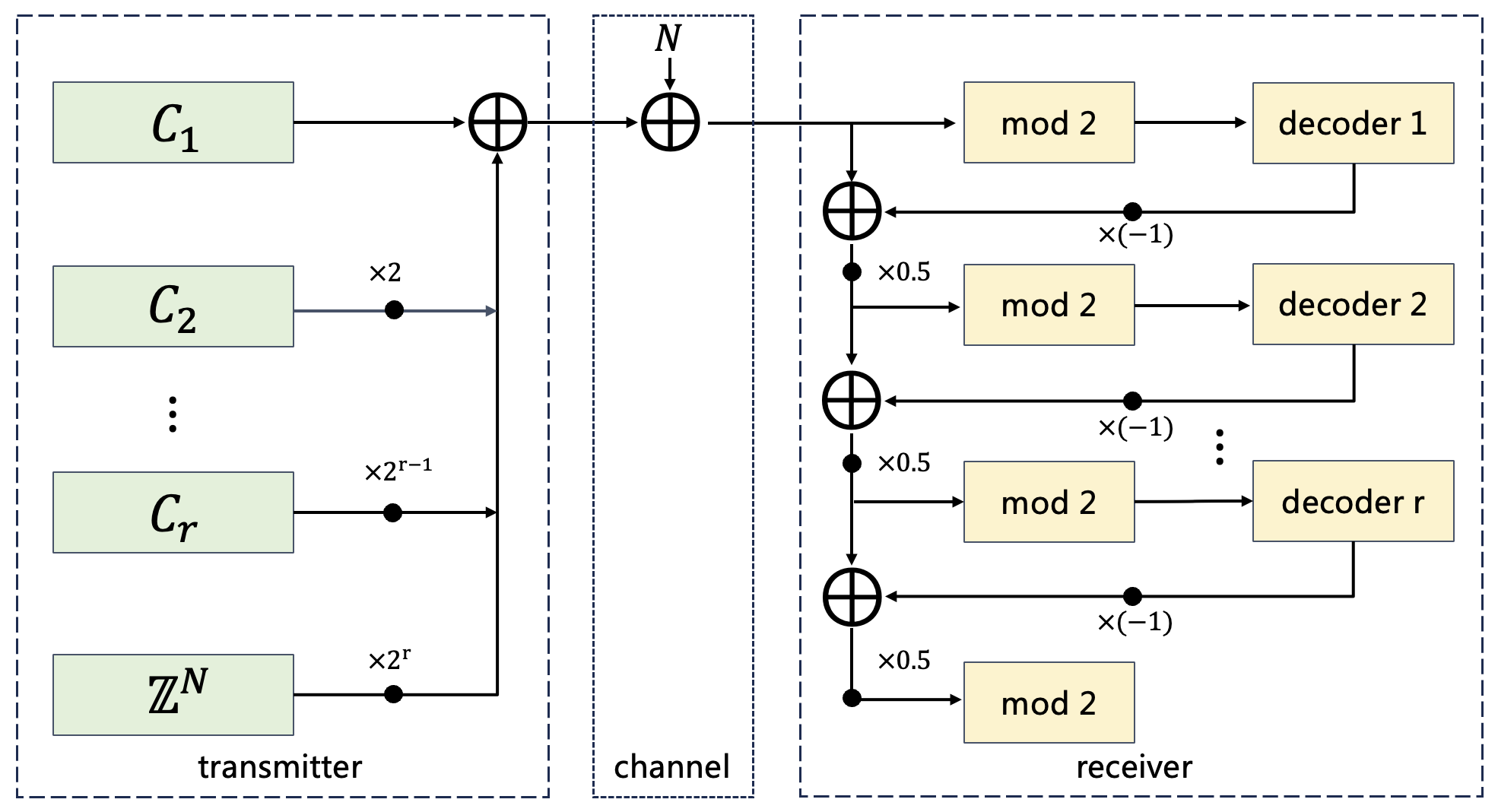}
    \caption{Framework of multilevel lattice encoding and decoding.}
    \label{fig:multilevel-decoding}
\end{figure}


\subsection{Decoding $H_{24}$}


For the principal sublattice $H_{24}$ of the Leech lattice $\Lambda_{24}$, its defining formula is given by
 \begin{align*}
     H_{24}=4\mathbb{Z}^{24}+2(24,23,2)+(24,12,8),
 \end{align*}
where (24, 23, 2) is a single-parity-check code and (24, 12, 8) denotes the extended Golay code.

For the decoding of $H_{24}$, only the receiver module depicted in Fig. \ref{fig:multilevel-decoding} needs to be considered, from which it can be readily concluded that $r=2$. Specifically, decoder 1 corresponds to the decoder for the (24, 12, 8) extended Golay code, and the proposed Parallel Decoding Method can be adopted for this stage; decoder 2 is designed for the (24, 23, 2) parity-check code. It calculates the soft information LLRs based on the input as illustrated in Fig. \ref{fig:multilevel-decoding} and performs hard decision. A parity-check is then conducted on the hard-decision results; if the check fails, the bit with the worst soft information (i.e., the smallest absolute LLR value) is flipped, and the decoded result is then output.

\subsection{Decoding $\Lambda_{24}$}

We use the rotated version $R\Lambda_{24}$ of Leech Lattices, and the following code formula: 
\begin{align*}
    R\Lambda_{24}=8\mathbb{Z}^{24}+4(24,23,2)+2(24,12,8)+(24,1,32)',
\end{align*}
where (24, 23, 2) denotes a parity-check code, (24, 12, 8) represents the extended Golay code and (24, 1, 32)$'$ corresponds to a repetition code. The last row of Leech Lattice's basis contains one ``$-3$'' and twenty-three ``$1$''s, which corresponds to a repetition code after the modulo-2 operation. Fig.~\ref{fig:Leech-matrix} shows the generator matrix of $R\Lambda_{24}$ adapted from \cite{conway1999sphere}.


\begin{figure}[ht]
    \centering
    \includegraphics[width=7.5cm]{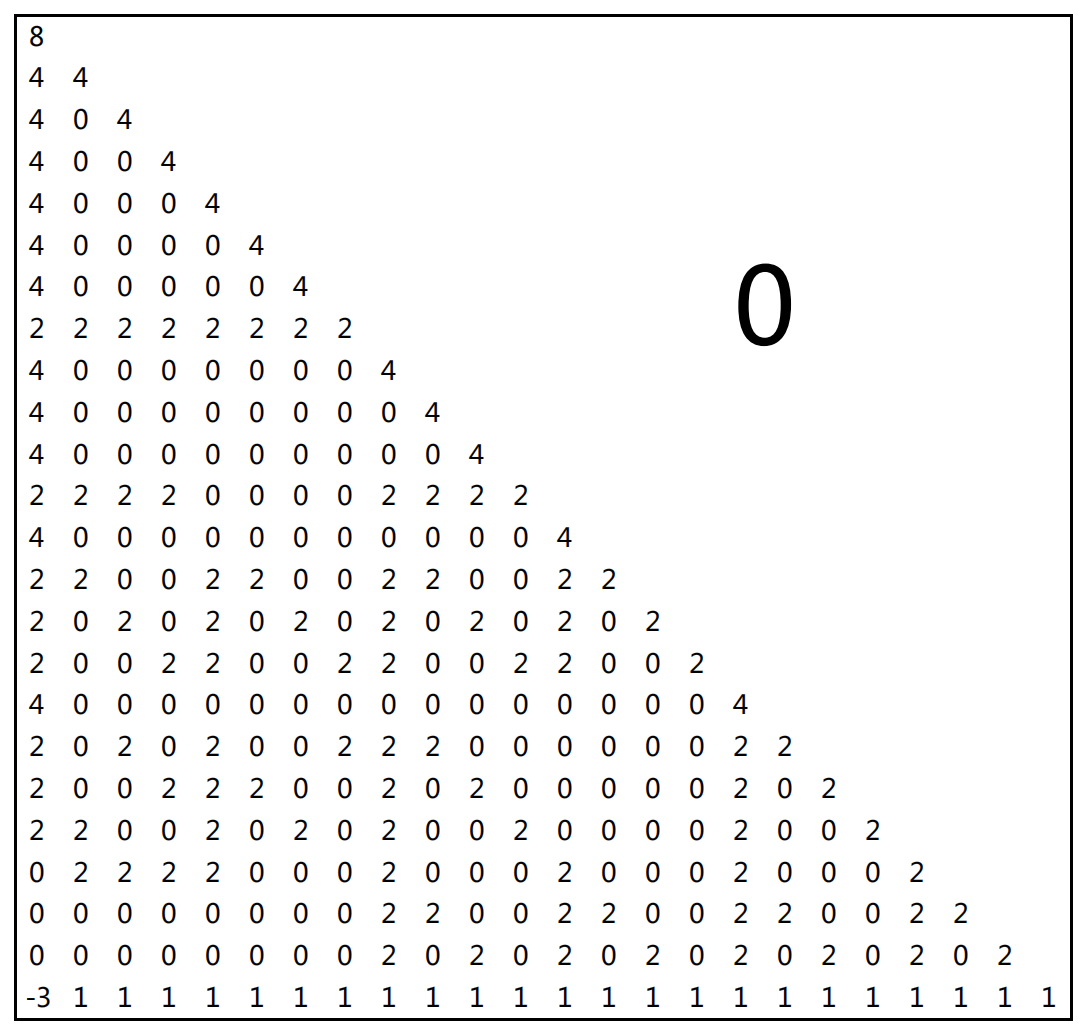}
    \caption{A generator matrix of $R\Lambda_{24}$.}
    \label{fig:Leech-matrix}
\end{figure}

For the decoding of $R\Lambda_{24}$, the same receiver module shown in Fig. \ref{fig:multilevel-decoding} applies, with the partition level readily determined as $r=3$. Specifically, decoder~1 targets the (24, 1, 32)$'$ repetition code and performs decoding via the total sum of LLRs; the detailed procedure is presented in Algorithm~\ref{alg:decoder1}. 
Decoder~2 corresponds to the (24, 12, 8) extended Golay code, where the parallel decoding method is employed; the process is detailed in Algorithm~\ref{alg:decoder2}. 
The matrix $G_{\text{golay}}$ is derived from the generator matrix $G_{\Lambda}$ of $R\Lambda_{24}$~\cite{conway1999sphere}, which is distinct from $G_{24}^{p,(i)}$ for $i=1,2,3$ used in the parallel decoding; hence, a conversion of the results is required. 
Specifically, the encoding relationship satisfies $u_1^{12} G_{24}^{p,(1)} = \hat{u}_1^{12} G_{\text{golay}}$, where $\hat{u}_1^{12}$ denotes the target message. 
To recover $\hat{u}_1^{12}$, we compute a right inverse $R_{\text{golay}}$ satisfying $G_{\text{golay}} R_{\text{golay}} = I_{12}$ through Gaussian elimination. 
Decoder~3 is designed for the (24, 23, 2) parity-check code; the decoding steps are given in Algorithm~\ref{alg:decoder3}. 
The matrix $G_{\text{pc}}$ extracted from $G_{\Lambda}$ is not in systematic form. 
It is transformed into systematic form by matrix $A$ such that $AG_{\text{pc}} = [I_{23}, \mathbf{p}]$, where $\mathbf{p}$ represents the last column of the product. 
The complete decoding procedure for $R\Lambda_{24}$ is described in Algorithm~\ref{alg:decoder-leech}. 
The notation $\lfloor \cdot \rceil$ indicates the rounding function in $\mathbb{Z}$, and the coordinate array of $\boldsymbol{\lambda}$ is defined in~\cite[Definition~3]{liu2025revisit}. Fig.~\ref{fig:multilevel-decoding(bler)} shows the BLER performance of the $R\Lambda_{24}$ using multilevel decoding (Algorithm~\ref{alg:decoder-leech}) with various list sizes $L$. 
The signal-to-noise ratio (SNR) is calculated as $10\log_{10}\left(\frac{1}{\sigma^2}\right)$~dB. $L=16$ suffices for our application, as larger values offer virtually no additional performance benefit.

\begin{algorithm}
\caption{Decoder 1 for the repetition code}
\label{alg:decoder1} 
\begin{algorithmic}
\Require $y_1^{24}$, $\sigma^2$, $G_{\text{rep}}$
\For{$1 \leq i \leq 24$}
    \State $LLR_i \gets f_{\sigma,2\mathbb{Z}}(y_i)$
\EndFor
\If{$\sum_{i=1}^{24}LLR_i \ge 0$}
    \State $u \gets 0$
\Else
    \State $u \gets 1$
\EndIf
\State $x_1^{24} \gets u \times G_{\text{rep}}$
\end{algorithmic}
\end{algorithm}

\begin{algorithm}
\caption{Decoder 2 for the extended Golay code}
\label{alg:decoder2}
\begin{algorithmic}
\Require $y_1^{24}$, $\sigma^2$, $L$, $G_{\text{golay}}$, $R_{\text{golay}}$, $G_{24}^{p,(1)}$

\For{$1 \leq i \leq 24$}
    \State $LLR_i \gets f_{\sigma,2\mathbb{Z}}(y_i)$
\EndFor
\State $u_1^{12} \gets \textbf{Parallel Decoder}         (LLR_1^{24}, L)$
\State $u_1^{12} \gets [u_1^{12} \times G_{24}^{p,(1)} \times R_{\text{golay}}] \ \text{mod} \ 2$
\State $x_1^{24} \gets u_1^{12} \times G_{\text{golay}}$
\end{algorithmic}
\end{algorithm}

\begin{algorithm}
\caption{Decoder 3 for the parity-check code}
\label{alg:decoder3}
\begin{algorithmic}
\Require $y_1^{24}$, $\sigma^2$, $G_{\text{pc}}$, $A$

\For{$1 \leq i \leq 24$}
    \State $LLR_i \gets f_{\sigma,2\mathbb{Z}}(y_i)$
    \If{$LLR_i \ge 0$}
        \State $x_i \gets 0$
    \Else
        \State $x_i \gets 1$
    \EndIf
\EndFor

\If{$[\sum_{i=1}^{24}x_i] \ \text{mod} \ 2 \ne 0 $}
    \State $index \gets \textbf{Min}(|LLR_1^{24}|)$
    \State $x_{index} \gets [x_{index}+1]\ \text{mod} \ 2 $
\EndIf
\State $u_1^{23} \gets [x_1^{23}\times A] \ \text{mod} \ 2$
\State $x_1^{24} \gets u_1^{23}\times G_{\text{pc}}$
\end{algorithmic}
\end{algorithm}








\begin{algorithm}
\caption{Decoder for the $R\Lambda_{24}$}
\label{alg:decoder-leech}
\begin{algorithmic}
\Require $\mathbf{y}$, $\sigma^2$, $L$, $G_{\Lambda}$

\State $G_{\text{rep}} \gets \text{rows of weight }1\text{ in }G_{\Lambda}$
\State $G_{\text{golay}} \gets \text{rows of weight }1,2\text{ in }G_{\Lambda}$  
\State $G_{\text{pc}} \gets \text{rows of weight }1,2,4\text{ in }G_{\Lambda}$

\State \textbf{Step 1: First-level decoding (Repetition code)}
\State $\bar{\mathbf{y}} \gets \mathbf{y} \bmod 2$  
\State $[\mathbf{u}_{1}, \mathbf{x}_1] \gets \textbf{Decoder 1}(\bar{\mathbf{y}}, \sigma^2, G_{\text{rep}})$ 
\State $\mathbf{y} \gets [\mathbf{y}-\mathbf{x}_1] / 2$

\State \textbf{Step 2: Second-level decoding (Golay code)}
\State $\bar{\mathbf{y}} \gets \mathbf{y} \bmod 2$
\State $[\mathbf{u}_2, \mathbf{x}_2] \gets \textbf{Decoder 2}\bigl(\bar{\mathbf{y}}, (\frac{\sigma}{2})^2, L, G_{\text{golay}}\bigr)$ 
\State $\mathbf{y} \gets [\mathbf{y}-2\mathbf{x}_2] / 2$

\State \textbf{Step 3: Third-level decoding (Parity-check code)}
\State $\bar{\mathbf{y}} \gets \mathbf{y} \bmod 2$
\State $[\mathbf{u}_3, \mathbf{x}_3] \gets \textbf{Decoder 3}\bigl(\bar{\mathbf{y}}, (\frac{\sigma}{4})^2, G_{\text{pc}}\bigr)$ 
\State $\mathbf{y} \gets [\mathbf{y}-4\mathbf{x}_3] / 2$

\State \textbf{Step 4: Final Estimation}
\State $\mathbf{z} \gets \lfloor \mathbf{y} \rceil \times G_{\Lambda}^{-1}$
\State $\boldsymbol{\lambda} \gets [\mathbf{u}_1, \mathbf{u}_2, \mathbf{u}_3, \mathbf{z}]$

\end{algorithmic}
\end{algorithm}

\begin{figure}[ht]
    \centering
    \includegraphics[width=7.5cm]{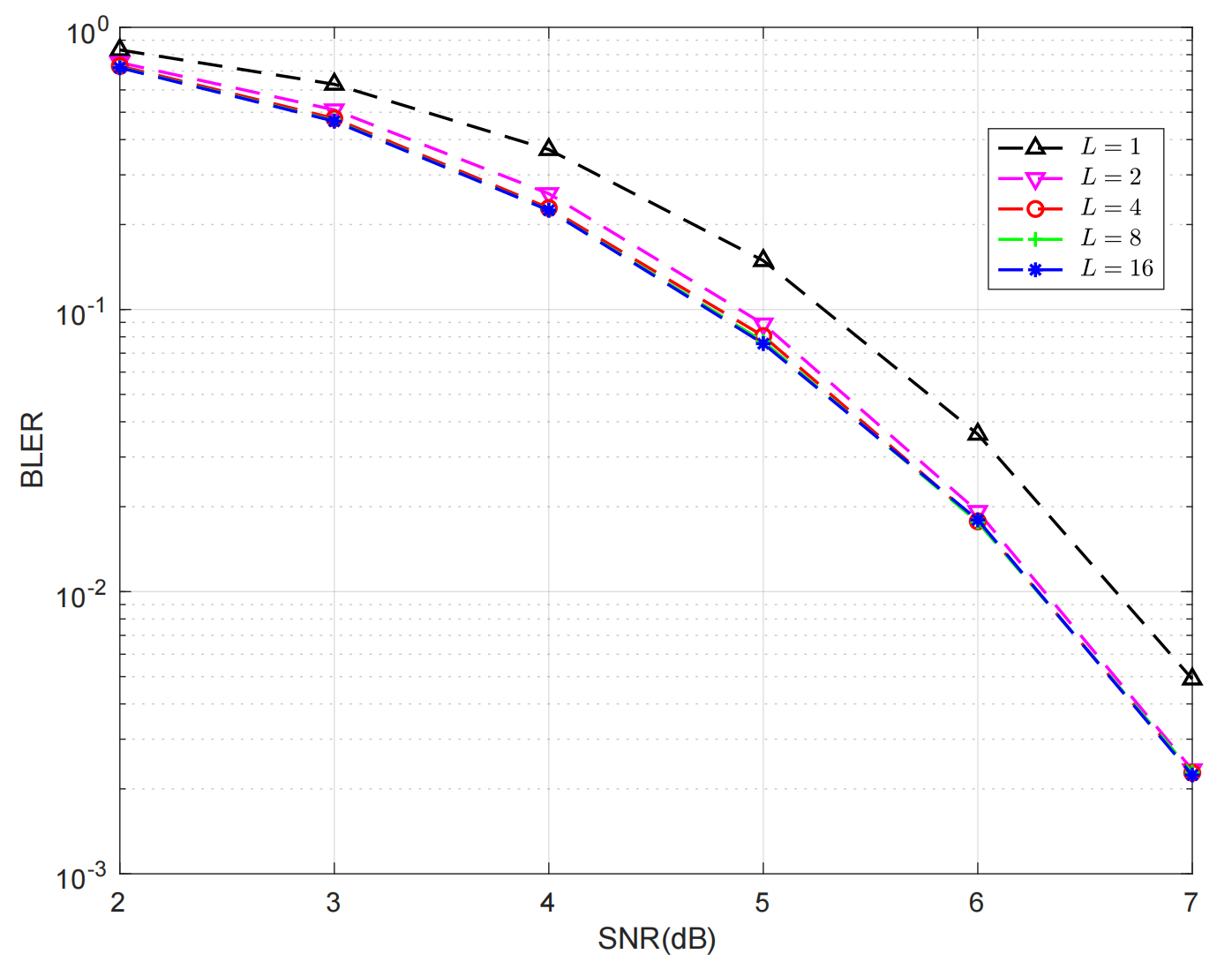}
    \caption{BLER performance of the $R\Lambda_{24}$ under multilevel decoding with different list sizes.}
    \label{fig:multilevel-decoding(bler)}
\end{figure}

In multilevel decoding, the primary challenge lies in the low-level stages, which face Gaussian noise with higher variance. Since layers are decoded sequentially, errors in early stages propagate to subsequent layers and degrade performance. For $R\Lambda_{24}$, an error in the bottom repetition code directly affects the Golay decoder. To improve performance, we introduce level lists between layers, denoted as $L_{\ell}=L_1\times L_2$. Specifically, we retain $L_1\in\{1,2\}$ paths at the first layer and $L_2$ paths at the second, while the third layer is used only for parity checking. Fig.~\ref{fig:multilevel-decoding(bler,level-list)} shows the BLER performance of the $R\Lambda_{24}$ using multilevel decoding with level list sizes $L_{\ell}$. 
With $L_1$ set to $2$, the decoding performance improves significantly. The performance curve nearly saturates as $L_2$ reaches $16$. Therefore, the recommended level list size is $L_{\ell}=2\times 16$. 

\begin{figure}[ht]
    \centering
    \includegraphics[width=7.5cm]{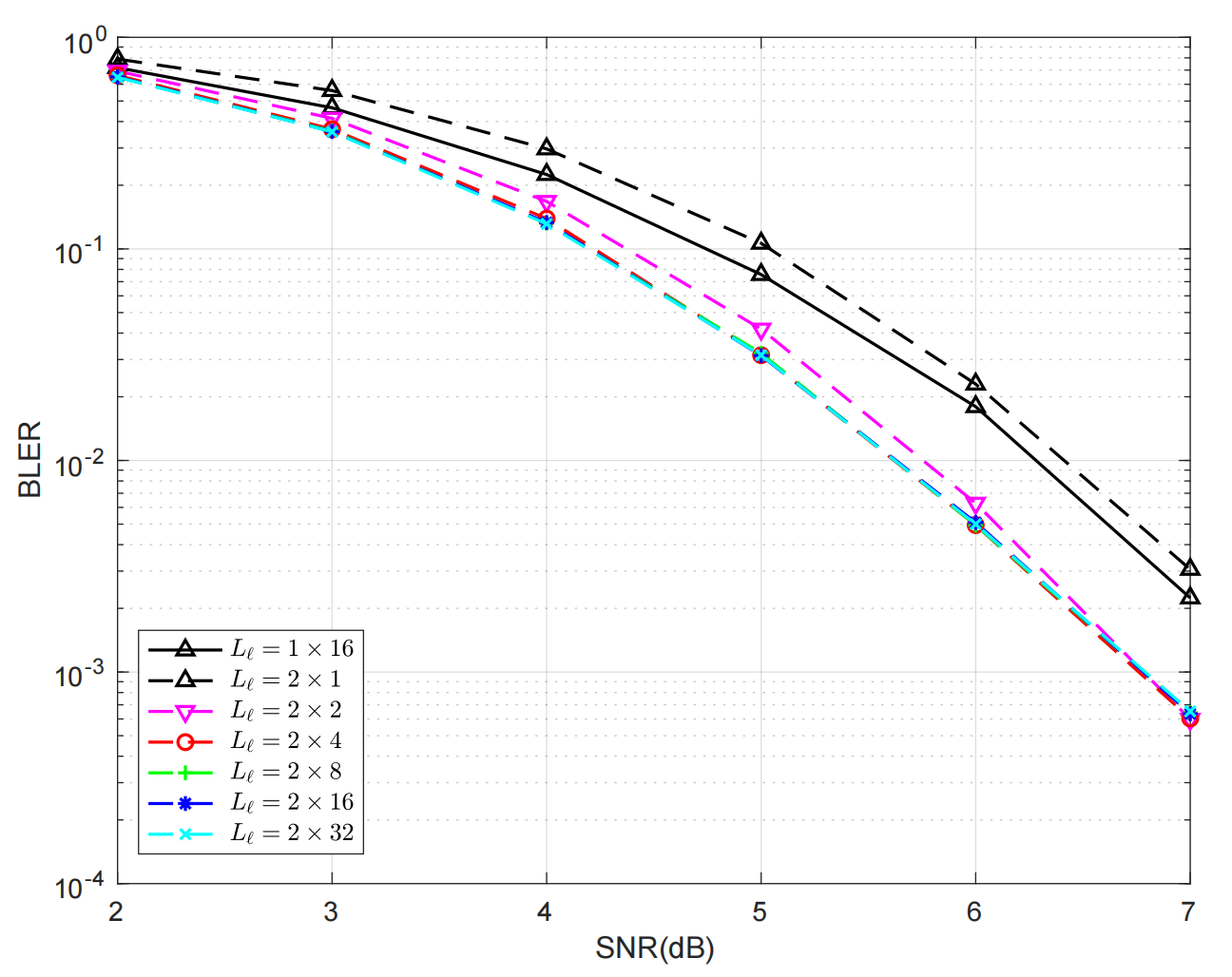}
    \caption{BLER performance of the $R\Lambda_{24}$ under multilevel decoding with different level list sizes $L_{\ell}=L_1\times L_2$.}
    \label{fig:multilevel-decoding(bler,level-list)}
\end{figure}

\section{Conclusion}
In this work, we identify several constructions of Golay codes from a PAC code perspective. Accordingly, SCL decoders for PAC codes can be applied to decode the corresponding Golay codes. Furthermore, based on three distinct PAC code constructions, we propose an efficient parallel list decoding algorithm. Simulation results demonstrate that this decoding algorithm can achieve near-maximum likelihood performance with only a small list size, and the decoding framework is significantly simplified because it eliminates the need for index permutation and codeword puncturing. Additionally, the proposed decoding method enables efficient decoding of related lattices like $\Lambda_{24}$ and $H_{24}$. Our future research direction is to conduct a detailed comparison of performance and complexity with existing decoding algorithms for Leech lattices and Golay codes \cite{vardy2002even}.





%
\bibliographystyle{IEEEtran}
\bibliography{ref}

\end{document}